%% file: onlinelp.tex
\documentclass[11pt,letterpaper,english]{article}

\usepackage{amsmath,amssymb}
\usepackage[ruled,vlined]{algorithm2e}
\usepackage{pgf, tikz}
\usetikzlibrary{arrows}
\usepackage{xspace, units}
\allowdisplaybreaks

\usepackage{lmodern}
\usepackage[T1]{fontenc}
\usepackage{textcomp}

\usepackage{nicefrac}

\usepackage{amsthm}
\usepackage{typearea}
\typearea{16}

\usepackage{color}
\definecolor{Darkblue}{rgb}{0,0,0.4}
\definecolor{Brown}{cmyk}{0,0.81,1.,0.60}
\definecolor{Purple}{cmyk}{0.45,0.86,0,0}
\usepackage[breaklinks]{hyperref}
\hypersetup{colorlinks=true,
            citebordercolor={.6 .6 .6},linkbordercolor={.6 .6 .6},
            citecolor=black,urlcolor=Darkblue,linkcolor=black,
            pdfpagemode=UseThumbs,pdfstartview=FitH}
\newcommand{\lref}[2][]{\hyperref[#2]{#1~\ref*{#2}}}

\newtheorem{theorem}{Theorem}

\newtheorem{lemma}[theorem]{Lemma}

\newtheorem{claim*}[theorem]{Claim}
\newtheorem{observation}[theorem]{Observation}

\newcommand{\e}{{\mathrm e}}

\renewcommand{\Pr}[1]{\mbox{\rm\bf Pr}\left[#1\right]}
\newcommand{\Ex}[1]{\mbox{\rm\bf E}\left[#1\right]}

\usepackage{todonotes}
\usepackage{authblk}

\newcommand{\OPT}{\mathrm{OPT}}
\newcommand{\growingmid}{\mathrel{}\middle|\mathrel{}}

\DeclareMathOperator*{\ALG}{ALG}

\title{Primal Beats Dual on Online Packing LPs in the Random-Order Model}
\date{}

\author[1]{Thomas Kesselheim\thanks{Supported by a fellowship within the Postdoc-Programme of the German Academic Exchange Service (DAAD).}}
\author[2]{Klaus Radke\thanks{Supported by the Studienstiftung des deutschen Volkes.}}
\author[2]{Andreas T\"onnis\thanks{Supported by the DFG GRK/1298 ``AlgoSyn''.}}
\author[2]{Berthold V\"ocking}
\affil[1]{Department of Computer Science, Cornell University, Ithaca, NY, USA.\authorcr kesselheim@cs.cornell.edu}
\affil[2]{Department of Computer Science, RWTH Aachen University, Germany.\authorcr \{radke, toennis, voecking\}@cs.rwth-aachen.de}

\begin{document}

\maketitle

\begin{abstract}
\input{abstract}

\end{abstract}

\thispagestyle{empty}
\setcounter{page}0
\clearpage

\section{Introduction}
\input{introduction}

\subsection{Model and Definitions}
\label{sec:definition}
\input{definitions}

\subsection{Our Results}
\input{results}

\subsection{Related Work}
\label{sec:related}
\input{related_work}

\section{A Robust Algorithm for Online Packing LPs}
\label{sec:algo}
\input{lp_algo}

\subsection{High Capacities}
\label{sec:high-b}
\input{lp_proof_high}

\subsection{Small Capacities}
\label{sec:small-b}
\input{lp_proof_small}

\section{Extensions and Variants}
\subsection{Truthfulness}
\input{truthful}

\subsection{Improved Bounds for Fixed Low Capacities}
\input{small_tailored}

\subsection{Online Generalized Assignment Problem}
\input{gap_proof}

\bibliographystyle{plain}
\bibliography{onlinelp}

\end{document}

%% file: abstract.tex

We study packing LPs in an online model where the columns are presented to the algorithm in random order. This natural problem was investigated in various recent studies motivated, e.g., by online ad allocations and yield management where rows correspond to resources and columns to requests specifying demands for resources. Our main contribution is a $1-O(\sqrt{\nicefrac{(\log{d})}{B}})$-competitive online algorithm, where $d$ denotes the {\em column sparsity}, i.e., the maximum number of resources that occur in a single column, and $B$ denotes the {\em capacity ratio $B$}, i.e., the ratio between the capacity of a resource and the maximum demand for this resource. In other words, we achieve a $(1 - \epsilon)$-approximation if the capacity ratio satisfies $B=\Omega(\frac{\log d}{\epsilon^2})$, which is known to be best-possible for any (randomized) online algorithms.

Our result improves exponentially on previous work with respect to the capacity ratio. In contrast to existing results on packing LP problems, our algorithm does not use dual prices to guide the allocation of resources over time. Instead, the algorithm simply solves, for each request, a scaled version of the partially known primal program and randomly rounds the obtained fractional solution to obtain an integral allocation for this request. We show that this simple algorithmic technique is not restricted to packing LPs with large capacity ratio of order $\Omega(\log d)$, but it also yields close-to-optimal competitive ratios if the capacity ratio is bounded by a constant. In particular, we prove an upper bound on the competitive ratio of $\Omega(d^{\nicefrac{-1}{(B-1)}})$, for any $B \ge 2$. In addition, we show that our approach can be combined with VCG payments and obtain an incentive compatible $(1-\epsilon)$-competitive mechanism for packing LPs with $B=\Omega(\frac{\log m}{\epsilon^2})$, where $m$ is the number of constraints. Finally, we apply our technique to the generalized assignment problem for which we obtain the first online algorithm with competitive ratio $O(1)$.

%% file: introduction.tex

Many optimization problems appear in the form of online problems, where the input is presented in a serial fashion. This means, requests come one at a time and have to be served directly and without knowledge about future requests. Most problems share the property that there is a set of resources with limited capacity. Any request, if served, raises a profit but it also consumes a certain amount of the resources. Whenever a request arrives, one has to make an irrevocable decision whether to serve it or not. The goal is to maximize the total profit of the served requests without exceeding the resource capacities.

These problems can be modeled by online packing linear programs, i.\,e.\ LPs with non-negative entries where the right-hand side is initially known while the variables appear online. Whenever a request arrives, one or possibly a set of columns of the constraint matrix are revealed along with their corresponding entries in the objective function. The algorithm has to choose at most one the current columns to allocate while not violating the packing constraints.

We study online packing LPs in the random-order model. Here, an adversary may generate an arbitrarily bad instance but he does not choose the order in which the requests arrive. Instead, the arrival order is chosen uniformly at random out of all possible permutations. This setting has been considered in various recent studies, most of which are motivated by online ad allocation. The best known online algorithms for this problem are based on the primal-dual method. They are $(1-\epsilon)$-competitive given that the  \emph{capacity ratio} $B$, i.\,e.\ the ratio between the capacity of a resource and the maximum demand for this resource, is large enough. For instances with $m$ resources and $n$ requests they require $B$ to be lower bounded by $\Omega(\frac{m}{\epsilon^2} \log \frac{n}{\epsilon})$ or $\Omega(\frac{m^2}{\epsilon^2} \log \frac{m}{\epsilon})$\footnote{We follow the convention of using $\log k$ to denote $\log k = \max\{\log_2 k, 1\}$ in asymptotic statements.}, respectively (see Section~\ref{sec:related} for details). These bounds are complemented by an instance with $B=\frac{\log m}{\epsilon^2}$ for which no online algorithm can achieve a $1-o(\epsilon)$-approximation, which leaves an exponential gap between the necessary and sufficient conditions on the capacity ratio for $(1-\epsilon)$-competitive algorithms.

In this paper, we present a $(1-\epsilon)$-competitive online algorithm for instances with capacity ratio $\Omega(\frac{\log d}{\epsilon^2})$, where $d \le m$ denotes the {\em column sparsity}, i.e., the maximum number of resources that occur in a single column. That is, we do not only close the exponential gap with respect to the capacity ratio for $(1-\epsilon)$-competitive algorithms in terms of $m$ but, additionally, achieve the first results on the capacity ratio in terms of $d$. We would like to point out that, apart from the capacity vector and the number of requests $n$, our algorithm does not need to be provided with any prior knowledge about the instance like, e.g., upper or lower bounds on the demands, the capacity ratio, or the coefficients in the objective function. In particular, given any instance with capacity ratio $B$ and sparsity $d$, the competitive ratio of the algorithm is $1-O(\sqrt{\nicefrac{(\log d)}{B}})$ without that the algorithm needs to be tuned with respect to these parameters.

%% file: definitions.tex

%

We consider packing LPs of the form $\max c^{T} x$ s.t.\ $Ax \leq b$ and $0\leq x\leq 1$ to model problems with $m$ resources and $n$ online requests coming in random order.
Each resource $i \in [m]$ has a capacity $b_i \geq 0$, which is initially known.
Additionally, we assume that the number of requests $n$ is known.
Every online request comes with a set of options, where each option has individual profit and resource consumptions.
That is, request $j \in [n]$ corresponds to variables $x_{j, 1}, \ldots, x_{j, K}$ and option $k \in [K]$ raises profit $c_{j,k} \geq 0$ while having resource consumption $a_{i, j, k} \geq 0$ for every resource $i$.
As only a single option can be selected, we additionally have the constraints $\sum_{k \in [K]} x_{j, k} \leq 1$.
The objective is to maximize the total profit without exceeding the resource capacities.
This can be written as the following linear program:
\begin{align*} \label{linearprogram:packing}
 \max          \quad \sum_{j \in [n]} \sum_{k \in [K]} c_{j, k} x_{j, k}&  \\
 \mbox{s.\,t.} \quad \sum_{j \in [n]} \sum_{k \in [K]} a_{i, j, k} x_{j, k}&\leq b_i & i\in [m] \nonumber \\
                                      \sum_{k \in [K]} x_{j, k}&\leq 1          & j\in [n] \nonumber                            
\end{align*}
Requests come in random order and in an online step multiple columns may appear simultaneously.

Our algorithms compute integral solutions to this LP, whose values will be denoted by $\ALG$. We compare these solutions to the fractional optimum, which we refer to by $\OPT$. We express the competitive ratio $\Ex{\ALG} / \OPT$ in terms of the capacity ratio $B = \min_{i \in [m]} \frac{b_i}{\max_{j \in [n], k \in [K]} a_{i,j,k}}$. 
Furthermore, we assume that every column in the constraint matrix $A$ has at most $d$ non-zero entries. By definition $d \leq m$.

Given a scaling factor $f > 0$ and a set of requests $S \subseteq [n]$, we will denote by $\mathcal{P}(f, S)$ the set of feasible solutions to the LP in which all $b_i$ values are scaled by $f$ and only requests from $S$ are served. Formally this is the set of all vectors $x$ such that $(A x)_i \leq f b_i$ for all $i \in [m]$, $0 \leq \sum_{k \in [K]} x_{j, k} \leq 1$ for all $j \in S$, and $x_{j, k} = 0$ for all $j \not\in S$.

%% file: results.tex

Our main contribution is a natural and robust algorithm for online packing LPs in the random-order model that is $1-O\big(\sqrt{\nicefrac{(\log{d})}{B}}\big)$-competitive. For the general case $d=m$ this matches the known lower-bound by Agrawal et al.~\cite{DBLP:journals/corr/abs-0911-2974}. Also for the case $d=1$ we match the lower bound by Kleinberg~\cite{DBLP:conf/soda/Kleinberg05}.

In each step the algorithm solves a scaled version of the revealed linear program and randomly rounds the solution to obtain an allocation.
In particular, when $\ell$ requests have been revealed we compute an optimal fractional solution of the linear program consisting of all visible columns where we set the capacity vector to $\nicefrac{\ell}{n} \cdot b$.
Then we interpret the fractional allocation of the current request as a probability distribution over its columns and randomly round it.
If the selected tentative option, together with previously allocated columns, does not violate the packing constraints then we allocate it permanently.

Compared to existing algorithms based on the primal-dual paradigm, an interesting advantage of this algorithm is that it does not require any kind of sampling phase in which no allocations are made. Hence, online requests do not suffer from the usual disadvantages of being at the beginning of the online sequence. Another advantage is that we only require very little information. While previous algorithms needed to know $B$ and $d$ upfront, ours flexibly adapts. We demonstrate this by analyzing the same algorithm in the case of small capacity ratios, i.\,e.\ for any $B \geq 2$, and show that it is $O\big(d^{-\nicefrac{2}{(B-1)}}\big)$-competitive. If, however, we know $B$ and $d$ in advance, we can improve this to $O\big(d^{-\nicefrac{1}{(B-1)}}\big)$.

A common motivation for online packing problems are online combinatorial auctions, in which bidders arrive online, report their valuations, and have to be served immediately. We show that with a slight modification our algorithm can be made truthful. That is, bidders do not have an incentive to misreport their valuation. We achieve competitive ratios $1-O\big(\sqrt{\nicefrac{(\log{m})}{B}}\big)$, $O\big(m^{-\nicefrac{2}{(B-1)}}\big)$, and $O\big(m^{-\nicefrac{1}{(B-1)}}\big)$ respectively.

We furthermore consider the online generalized assignment problem in the random-order model, which is the special case of online packing linear programs with $d=1$. For this problem we present a $\nicefrac{1}{8.1}$-competitive algorithm that does not need any assumptions on the parameters of the input instance. Although the analysis serves mainly for the purpose of illustrating our general proof technique without using any technical Chernoff bounds, it is the first result that covers weighted bipartite matching, AdWords and knapsack simultaneously.

%% file: related_work.tex

The work on online packing linear programs was initiated by Buchbinder and Naor~\cite{DBLP:journals/mor/BuchbinderN09} who analyzed the worst-case model. They assumed that the capacity vector is initially known while the columns, i.\,e.\ the variables, are presented one at a time by an adversary. They presented an optimal primal-dual based algorithm that obtains a competitive ratio of $O\big(1 / (\log m + \log \max_{i \in [m]}\frac{a_i(\max)}{a_i(\min)})\big)$. Here, $a_i(\max)$ and $a_i(\min)$ are the largest, respectively the smallest, non-zero entry in row $i$ of the constraint matrix.

In recent years, the research on online packing LPs was focused on the less pessimistic random-order model where it is possible to obtain $(1-\epsilon)$-competitive algorithms. In the existing work it is generally assumed that the capacity ratios are large and that this ratio is known at the beginning. The first results were independently presented by Feldman et al.~\cite{DBLP:conf/esa/FeldmanHKMS10} and by Agrawal et al.~\cite{DBLP:journals/corr/abs-0911-2974}, who all built on techniques developed by Devanur and Hayes~\cite{DBLP:conf/sigecom/DevanurH09} for the AdWords problem. Their model already allowed up to $K$ columns to arrive in every online step. The algorithm by Feldman et al.\ is $(1-\epsilon)$-competitive when $B = \Omega\big(\frac{m \log{(n K)} }{ \epsilon^3}\big)$ and $\OPT = \Omega\big(\frac{c_{\max} m \log{(n K)} }{\epsilon}\big)$, where $c_{\max}$ is the largest entry in the objective function vector. Agrawal et al.\ needed the assumption $B = \Omega\big(\frac{m \log{(\nicefrac{n K}{\epsilon})}}{\epsilon^2}\big)$ or $\OPT = \Omega\big(\frac{c_{\max}m^2 \log{(\nicefrac{n}{\epsilon})}}{\epsilon^2}\big)$. Additionally, the second paper provided a lower bound on $B$ of $B = \Omega\big(\frac{\log m}{\epsilon^2}\big)$ to allow for $(1-\epsilon)$-competitive algorithms. This bound is matched by our current result. Later, Molinaro and Ravi~\cite{DBLP:conf/icalp/MolinaroR12} presented a $(1-\epsilon)$-competitive algorithm assuming $B= \Omega\big(\frac{m^2 \log{(\nicefrac{m}{\epsilon})}}{\epsilon^2}\big)$ and thus removed the dependence on the number of requests $n$.

Online packing LPs have also been analyzed in a generalized i.i.d.\ model by Devanur et al.~\cite{DBLP:conf/sigecom/DevanurJSW11} who obtained a $(1-\epsilon)$-competitive algorithm when $B= \Omega\big(\frac{\log{(\nicefrac{m}{\epsilon})}}{\epsilon^2}\big)$ as well as a lower bound of $B= \Omega\big(\frac{\log{m}}{\epsilon^2}\big)$. However, this model is significantly weaker than the random-order model as discussed in detail by Molinaro and Ravi~\cite{DBLP:conf/icalp/MolinaroR12}. Besides, like all previously mentioned algorithms and in contrast to ours, also this model requires exact knowledge about $B$. Furthermore, these algorithms apply a primal-dual scheme with multiplicative increases. Ours instead uses a fundamentally different approach and can be combined with any technique to solve LPs. Finally, as already pointed out, our bounds depend on $d$ rather then $m$, which means we get significantly better bounds in sparse matrices.

Prior to the general packing programs, other more specialized allocation problems have been investigated in the random-order online model. Most of them are subsumed by the generalized assignment problem, which we analyze, too. Many are themselves generalizations of the secretary problem, which intrinsically assumes the random-order online model. An optimal $\nicefrac{1}{\e}$-competitive algorithm is known since the works of Lindley~\cite{lindley1961dynamic} and Dynkin~\cite{dynkin1963optimum}. If one allows to choose $k$ elements there is a $(1 - O(\sqrt{\nicefrac{1}{k}}))$-competitive algorithm by Kleinberg~\cite{DBLP:conf/soda/Kleinberg05}. A generalization with combinatorial flavor is the edge-weighted matching problem. Here, an optimal $\nicefrac{1}{\e}$-competitive algorithm by Kesselheim et al.~\cite{DBLP:conf/esa/KesselheimRTV13} is known that built on the work of Korula and P\'al~\cite{DBLP:conf/icalp/KorulaP09}. Another generalization that corresponds to linear programs with only one constraint is the online knapsack problem, where the best known competitive ratio is $\nicefrac{1}{10 \e}$ by Babaioff et al.~\cite{DBLP:conf/approx/BabaioffIKK07}. Our result for the generalized assignment problem improves on this ratio.

An important special case of the generalized assignment problem is the one where resource consumptions and profits are identical. In this case there are competitive algorithms in the worst-case model, even without assumptions on capacities. This line of research was initiated by Karp et al.~\cite{DBLP:conf/stoc/KarpVV90}, who gave an optimal $1-\nicefrac{1}{\e}$-competitive algorithm for unweighted bipartite matching in the worst-case online model. Mehta et al.~\cite{DBLP:journals/jacm/MehtaSVV07} introduced the AdWords problem and also presented a $1-\nicefrac{1}{\e}$-competitive algorithm in the worst-case online model. These problems were first considered in the random-order model by Goel and Mehta~\cite{DBLP:conf/soda/GoelM08}. Here, the best known competitive ratio for bipartite matching is $0.696$ by Mahdian and Yan~\cite{DBLP:conf/stoc/MahdianY11} (see also Karande et al.~\cite{DBLP:conf/stoc/KarandeMT11}). For the AdWords problem Devanur and Hayes~\cite{DBLP:conf/sigecom/DevanurH09} gave a primal-dual based $(1-\epsilon)$-competitive algorithm when the capacities are large. In the case of general capacities the best known competitive ratio for the AdWords problem in the random-order model is \nicefrac{1}{2}~\cite{DBLP:journals/jacm/MehtaSVV07}.

%% file: lp_algo.tex

We consider the following natural online algorithm: 
When the $\ell$th request, say with index $j$, arises, we compute the optimal fractional solution for all requests that we have seen up to this point with capacities scaled by a factor of $\frac{\ell}{n}$. 
Then the fractional allocation of request $j$ is used as guide to determine which column to choose by interpreting it as a probability distribution over the request's options.
The selected tentative allocation is carried out if the resulting solution is feasible with respect to the unscaled capacities. 
If it is not, the request is discarded.

\begin{algorithm}
\caption{Online packing LP\label{alg:packing_algorithm}}
\SetKwInOut{Input}{Input} \SetKwInOut{Output}{Output}
Let $S$ be the index set of known requests, initially $S := \emptyset$\;
Set $y := 0$\;
\For(\tcp*[f]{steps $\ell = 1$ to $n$}){\emph{each arriving request} $j$}{
  Set $S := S \cup \{j\}$ and $\ell := |S|$\;
  Let $\tilde{x}^{(\ell)}$ be an optimal solution of the scaled LP $\max_{x \in \mathcal{P}(\frac{\ell}{n}, S)} c^T x$\;
  Choose an option $k^{(\ell)}$ (possibly none) where option $k$ has probability $\tilde{x}^{(\ell)}_{j, k}$; \tcp*[f]{rand.\ rounding}\\
  Define $x^{(\ell)}$ with $x^{(\ell)}_{j', k} = \begin{cases} 1 , & \text{if } j' = j \text{ and } k = k^{(\ell)};  \\
                                                                0 , & \text{otherwise} ;\end{cases}$ \tcp*[f]{tentative allocation}\\
  \If(\tcp*[f]{feasibility test}){$A(y+x^{(\ell)}) \leq b$}{
    Set $y := y+x^{(\ell)}$; \tcp*[f]{permanent online allocation}\\
  }
}
\end{algorithm}

Observe that this algorithm is invariant with respect to scaling constraints. Therefore, we can assume without loss of generality that $\max_{j \in [n], k \in [K]} a_{i, j, k} = 1$ for every constraint $i \in [m]$. In this case $B = b_{\min} := \min_{i \in [m]} b_i$.

In the following sections we analyze the above algorithm for large capacities $b_{\min} = \Omega\left(\frac{\log d}{\epsilon^2}\right)$ and small capacities $b_{\min} \geq 2$.
On a high level both proofs follow the same approach. First, for every round $\ell$, we bound the expected value of the locally optimal solution $\tilde{x}^{(\ell)}$. We exploit that this optimal solution is independent of the order of all known elements up to this point. Under these circumstances, we can interpret the current request as drawn uniformly at random from all known requests. This way we bound the expected value contributed to the solution of the tentative allocation. We furthermore obtain a bound on the probability that this allocation is feasible. This bound only uses the random order of the first $\ell - 1$ requests, which was irrelevant up to this point.

%% file: lp_proof_high.tex

In this section we show that the algorithm achieves a $1 - O\left( \sqrt{\frac{\log d}{b_{\min}}} \right)$ approximation. 
In other words, to get a $(1 - \epsilon)$-approximation, we require $b_{\min} = \Omega\left(\frac{\log d}{\epsilon^2}\right)$. 
However, also in case that $b_{\min}$ does not fulfill this bound for any $\epsilon < 1$, the algorithm still has near optimal performance guarantees. 
We cover these cases in Section~\ref{sec:small-b}.

Our analysis will proceed in three larger steps. First, we consider the value of the optimal solution of the scaled LP (Lemma~\ref{lemma:local_solution_value}). 
This way, for each round $\ell$, we can bound what the tentative selection in round $\ell$ would contribute to the objective function. 
Second, we obtain a bound on the probability that the capacity for a certain constraint is exhausted by round $\ell$ (Lemma~\ref{lemma:allocation_probability}). 
This allows us to estimate the probability that a tentative allocation can be carried out. 
Finally, we add up the expected value obtained in all rounds $\ell$.

\begin{lemma} \label{lemma:local_solution_value}
Let $S \subseteq [n]$ be a random subset of requests with $\lvert S \rvert = \ell$ and $\ell \geq 2 \sqrt{\frac{1 + \ln d}{b_{\min}}} n$.
Then we have
\[
\Ex{\max_{x \in \mathcal{P}(\frac{\ell}{n}, S)} c^T x} 
\geq \left(1 - 9 \sqrt{\frac{1 + \ln d}{\frac{\ell}{n} b_{\min}}} \right) \frac{\ell}{n} \cdot \max_{x \in \mathcal{P}(1, [n])} c^T x \enspace .
\]
\end{lemma}

\begin{proof}
Let $x^\ast$ be an optimal solution of the full unscaled LP, i.\,e.\ of $\max_{x \in \mathcal{P}(1, [n])} c^T x$.
Project it to the set of relevant requests by setting $x'_{j, k} = x^\ast_{j, k}$ if $j \in S$ and $x'_{j, k} = 0$ otherwise. 
Obviously, we have $\Ex{x'_{j,k}} = \frac{\ell}{n} x^\ast_{j,k}$.
For each $(j, k) \in [n] \times [K]$, let $C_{j,k} = \{ i \in [m] \mid a_{i,j,k} > 0\}$ be the set of constraints that are influenced by the variable $x_{j, k}$.
Furthermore, we define for each variable a scaling factor $F_{j, k} = \min \left\{1,  \min_{i \in C_{j, k}} \frac{\nicefrac{\ell}{n} \cdot b_i}{(A x')_i} \right\}$ and set $x''_{j,k} = x'_{j,k} F_{j,k}$.
Now $x''$ satisfies $A x'' \leq \frac{\ell}{n} b$ since 
$(Ax'')_i = \sum_{j,k} a_{i,j,k} x'_{j,k} F_{j,k} \leq \sum_{j,k} a_{i,j,k} x'_{j,k} \frac{\nicefrac{\ell}{n} \cdot b_i}{(A x')_i} = \frac{\ell}{n} b_i$. 
In the rest of the proof we will analyze the expected value of the scaling factors $F_{j,k}$ in order to show 
$\Ex{x''_{j,k}} \geq \left(1 - 9 \sqrt{\frac{1 + \ln d}{\nicefrac{\ell}{n} \cdot b_{\min}}}\right) \Ex{x'_{j,k}}$, which gives the result.

Fix a request $\tilde{j} \in [n]$. We first turn our considerations to the conditional probability space in which $\tilde{j} \in S$. Furthermore, fix a constraint $i \in [m]$. The random variable $(Ax')_i$ is a sum of random variables from $[0, 1]$. Let us show that it is possible to apply a Chernoff bound, even though the involved random variables are not independent. For this purpose, we define $X_j = \sum_{k \in [K]} a_{i, j, k} x_{j, k}'$ for each request and show that these random variables are $1$-correlated in the sense of \cite{Panconesi1997}. For each $X_j$, $j \in [n] \setminus \{ \tilde{j} \}$, we define the respective twin variable $\hat{X}_j$ to be set to $\sum_{k \in [K]} a_{i, j, k} x_{j, k}^\ast$ with probability $\frac{\ell}{n}$ and $0$ otherwise. This obviously yields $\Ex{X_j \growingmid \tilde{j} \in S} \leq \Ex{\hat{X}_j \growingmid \tilde{j} \in S}$. Let $I \subseteq [n] \setminus \{ \tilde{j} \}$ and for $j \in I$ let $s_j$ be an arbitrary positive integer. We have
\begin{align*}
\Ex{ \prod_{j \in I} X_j^{s_j} \growingmid \tilde{j} \in S} 
& = \Ex{ \prod_{j \in I} \left(\sum_{k \in [K]} a_{i, j, k} x_{j, k}'\right)^{s_j} \growingmid \tilde{j} \in S} = \left( \prod_{j \in I} \left(\sum_{k \in [K]} a_{i, j, k} x_{j, k}^\ast \right)^{s_j} \right) \Pr{S \supseteq I \growingmid \tilde{j} \in S} \\
& = \left( \prod_{j \in I} \left(\sum_{k \in [K]} a_{i, j, k} x_{j, k}^\ast \right)^{s_j} \right) \frac{\binom{n - \lvert I \rvert - 1}{\ell - \lvert I \rvert - 1}}{\binom{n - 1}{\ell - 1}} \leq \left( \prod_{j \in I} \left(\sum_{k \in [K]} a_{i, j, k} x_{j, k}^\ast \right)^{s_j} \right) \left( \frac{\ell}{n} \right)^{\lvert I \rvert} \\
& = \prod_{j \in I}\left( \left(\sum_{k \in [K]} a_{i, j, k} x_{j, k}^\ast \right)^{s_j} \frac{\ell}{n}\right) 
= \prod_{j \in I} \Ex{\hat{X}_j^{s_j} \growingmid \tilde{j} \in S} \enspace,
\end{align*}
which shows that the random variables $X_j$ are $1$-correlated.

Furthermore, $\Ex{(Ax')_i \growingmid \tilde{j} \in S} \leq \frac{\ell}{n} b_i + 1$ because $x^\ast$ is a feasible LP solution and $a_{i, j, k} \leq 1$.
Therefore, we can apply a Chernoff bound on $(Ax')_i$ to get
\[
\Pr{(Ax')_i \geq (1 + \delta) \left( \frac{\ell}{n} b_i + 1\right) \growingmid \tilde{j} \in S} \leq \exp\left( - \frac{\delta^2}{3} \left( \frac{\ell}{n} b_i + 1\right) \right) \leq \exp\left( - \frac{\delta^2}{3}  \frac{\ell}{n} b_{\min} \right) \enspace .
\]

Observe that $\frac{\nicefrac{\ell}{n} \cdot b_i}{(Ax')_i} \leq \frac{1}{1 + \delta + \nicefrac{1}{\sqrt{b_{\min}}}}$ implies $(Ax')_i \geq (1 + \delta) \left( \frac{\ell}{n} b_i + 1\right)$ where we use $\ell \geq \frac{2n}{\sqrt {b_{\min}}}$. 
Applying the definition of $F_{\tilde{j}, k}$ and a union bound we obtain
\begin{align}
\Pr{F_{\tilde{j}, k} \leq \frac{1}{1 + \delta + \frac{1}{\sqrt{b_{\min}}}} \growingmid \tilde{j} \in S} 
& \leq \sum_{i \in C_{\tilde{j}, k}} \Pr{\frac{\frac{\ell}{n} b_i}{(Ax')_i} \leq \frac{1}{1 + \delta + \frac{1}{\sqrt{b_{\min}}}}\growingmid \tilde{j} \in S} \nonumber \\
& \leq \sum_{i \in C_{\tilde{j}, k}} \Pr{(Ax')_i \geq (1 + \delta) \left( \frac{\ell}{n} b_i + 1\right) \growingmid \tilde{j} \in S} \nonumber \\
& \leq d \exp\left( - \frac{\delta^2}{3}  \frac{\ell}{n} b_{\min} \right) \enspace .  \label{ineq:main_chernoff}
\end{align}
Set the step width $\xi = \sqrt{3 \cdot \frac{1 + \ln d}{\nicefrac{\ell}{n} \cdot b_{\min}}}$. 
We can bound the expectation of $F_{\tilde{j}, k}$ by using
\begin{align*}
\Ex{F_{\tilde{j}, k} \growingmid \tilde{j} \in S} 
& \geq \sum_{i=0}^\infty \frac{1}{1 + (i+1)\xi} \cdot \Pr{\frac{1}{1 + (i+1)\xi} < F_{\tilde{j}, k} \leq \frac{1}{1 + i\xi} \growingmid \tilde{j} \in S} \\
& = \sum_{i=0}^\infty \left( 1 - \frac{(i+1)\xi}{1 + (i+1)\xi} \right) \cdot \Pr{\frac{1}{1 + (i+1)\xi} < F_{\tilde{j}, k} \leq \frac{1}{1 + i\xi} \growingmid \tilde{j} \in S} \\
& = 1 - \sum_{i=0}^\infty \frac{(i+1)\xi}{1 + (i+1)\xi} \cdot \Pr{\frac{1}{1 + (i+1)\xi} < F_{\tilde{j}, k} \leq \frac{1}{1 + i\xi} \growingmid \tilde{j} \in S} \enspace .
\end{align*}
To bound the sum, we split it into three parts, consisting of the ranges $i \in \{0, 1, 2\}$,
$i \in \left\{3, \ldots, \left\lfloor \nicefrac{1}{\xi} + 1\right\rfloor \right\}$ and $i \in \left\{\left\lfloor \nicefrac{1}{\xi} + 1\right\rfloor + 1, \ldots \right\}$.

The sum for $i \in \{0, 1, 2\}$ can be bounded by
\[
\sum_{i \in \{0, 1, 2\}}\frac{(i+1)\xi}{1 + (i+1)\xi} \cdot \Pr{\frac{1}{1 + (i+1)\xi} < F_{\tilde{j}, k} \leq \frac{1}{1 + i\xi} \growingmid \tilde{j} \in S} \leq 3 \xi \enspace ,
\]
where we used that the sum of probabilities is at most one and the largest term is the one for $i=2$.

In case $i \in \left\{3, \ldots, \left\lfloor \nicefrac{1}{\xi} + 1\right\rfloor \right\}$, we set $\delta = (i-1)\xi$. 
Since this implies $\delta \in (0, 1]$ and we have $\xi \geq \frac{1}{\sqrt{b_{\min}}}$, inequality (\ref{ineq:main_chernoff}) gives us
\[
\Pr{F_{\tilde{j}, k} \leq \frac{1}{1 + i\xi} \growingmid \tilde{j} \in S} 
\leq d \exp\left( - \frac{(i-1)^2 \xi^2}{3}  \frac{\ell}{n} b_{\min} \right) 
\leq d \exp\left( - (i-1)^2 (1 + \ln d) \right) 
\leq \exp\left( -i+1 \right) \enspace .
\]
Using this to bound the second sum we obtain
\begin{align*}
\sum_{i = 3}^{\left\lfloor \frac{1}{\xi} + 1\right\rfloor} &\frac{(i+1)\xi}{1 + (i+1)\xi} \cdot \Pr{\frac{1}{1 + (i+1)\xi} < F_{\tilde{j}, k} \leq \frac{1}{1 + i\xi} \growingmid \tilde{j} \in S} \\
& \leq \sum_{i = 3}^{\left\lfloor \frac{1}{\xi} + 1 \right\rfloor} (i+1)\xi \cdot \Pr{F_{\tilde{j}, k} \leq \frac{1}{1 + i\xi} \growingmid \tilde{j} \in S} 
 \leq \sum_{i = 3}^{\left\lfloor \frac{1}{\xi} + 1\right\rfloor} (i + 1)\xi \cdot \exp\left( -i+1 \right) 
 \leq \frac{4 \e - 3}{\e (\e-1)^2} \xi \leq \xi \enspace .
\end{align*}

Finally, for $i \in \left\{\left\lfloor \nicefrac{1}{\xi} + 1\right\rfloor + 1, \ldots \right\}$ we combine the terms to use inequality (\ref{ineq:main_chernoff}) only once, and obtain
\begin{align*}
\sum_{i = \left\lfloor \frac{1}{\xi} + 1\right\rfloor + 1}^{\infty} &\frac{(i+1)\xi}{1 + (i+1)\xi} \cdot \Pr{\frac{1}{1 + (i+1)\xi} < F_{\tilde{j}, k} \leq \frac{1}{1 + i\xi} \growingmid \tilde{j} \in S}\\ 
&\leq \Pr{F_{\tilde{j}, k} \leq \frac{1}{1 + \left( \left\lfloor \frac{1}{\xi} + 1\right\rfloor + 1 \right)\xi} \growingmid \tilde{j} \in S} 
 \leq \Pr{F_{\tilde{j}, k} \leq \frac{1}{2 + \frac{1}{\sqrt{b_{\min}}}} \growingmid \tilde{j} \in S} \\
&\leq d \exp\left( - \frac{1}{3}  \frac{\ell}{n} b_{\min} \right) 
 \leq \exp\left( - \frac{1}{6}  \frac{\ell}{n} b_{\min} \right) \leq \xi \enspace ,
\end{align*}
where we used that $\frac{\ell}{n} b_{\min} \geq 6 \ln d$.

Combining the three bounds we obtain $\Ex{F_{\tilde{j}, k} \growingmid \tilde{j} \in S} \geq 1 - 5\xi \geq 1 - 9  \sqrt{ \frac{1 + \ln d}{\nicefrac{\ell}{n} \cdot b_{\min}}}$.
This gives the claimed $\Ex{x''_{\tilde{j}, k}} \geq \left(1 - 9  \sqrt{ \frac{1 + \ln d}{\nicefrac{\ell}{n} \cdot b_{\min}}} \right) \Ex{x'_{\tilde{j}, k}}$, since $x_{\tilde{j}, k}' = 0$ for $\tilde{j} \notin S$.
\end{proof}

A tentative allocation can only be made permanent if the remaining capacity in every involved constraint is high enough. In the next lemma, we will show that in most rounds, this is the case with sufficiently high probability. Instead of bounding the consumption of the actual allocation, we will consider the previous tentative allocations. It is relatively easy to see that in expectation their consumption is not too high. To get the probability bound, we apply a Chernoff bound. However, we need to be very careful here as the tentative allocation in a round is obviously correlated to tentative allocations in previous rounds. Fortunately, we are able to show that even conditioned on outcomes in later rounds, the randomization in earlier rounds can still be considered unbiased. A simpler version of this kind of argument has already been applied in \cite{DBLP:conf/esa/KesselheimRTV13}.

\begin{lemma} \label{lemma:allocation_probability}
Consider round $\ell$ with $9 \sqrt{\frac{1 + \ln d}{b_{\min}}}  n \leq \ell \leq \left ( 1 - 9 \sqrt{\frac{1 + \ln d}{b_{\min}}}\right)  n$. 
Let $S \subseteq [n]$, $\lvert S \vert = \ell - 1$, be any set of $\ell - 1$ requests and let $\mathcal{E}_S$ be the event that the requests in $S$ come within the first $\ell - 1$ steps of the random input order.
Then, conditioned on $\mathcal{E}_S$, the sum of previous tentative allocations violate any constraint $i$ with probability at most
\[
\Pr{{\left(\sum_{\ell' < \ell} A x^{(\ell')}\right)_i} > b_i - 1 \growingmid \mathcal{E}_S} \leq \frac{1}{d} \exp\left( - \frac{n - \ell}{n} \sqrt{b_i} \right) \enspace .
\]
\end{lemma}

\begin{proof}
We will first argue that the involved random variables $X_{\ell'} := (A x^{(\ell')})_i$ allow applying a Chernoff bound, even though they are not independent. 
For this purpose, we follow the approach in \cite{Panconesi1997} and show that they are $1$-correlated. 
For each $X_{\ell'}$, $1 \leq \ell' \leq \ell - 1$, we define the respective twin variable $\hat{X}_{\ell'}$ to be set to $1$ with probability $\frac{b_i}{n}$ and $0$ otherwise. Now we need to show that for each set $I \subseteq [\ell - 1]$ and positive integers $s_{\ell'}$, $\ell' \in I$, we have
\[
\Ex{\prod_{\ell' \in I} X_{\ell'}^{s_{\ell'}} \growingmid \mathcal{E}_S } \leq \prod_{\ell' \in I} \frac{b_i}{n} \enspace.
\]
We show this claim by induction on $\lvert I \rvert$. For $\lvert I \rvert = 0$, the statement is trivially true. 
So let us consider the case that $I = \{ \ell'_1 \} \cup I'$ with $\ell'_1 < \ell'$ for all $\ell' \in I'$. 
By induction hypothesis, we already know
\[
\Ex{\prod_{\ell' \in I'} X_{\ell'}^{s_{\ell'}} \growingmid \mathcal{E}_S } \leq \prod_{\ell' \in I'} \frac{b_i}{n} \enspace.
\]
Let $a \geq 0$ be an arbitrary real number and consider the conditional probability space in which not only $\mathcal{E}_S$ but also $\prod_{\ell' \in I'} X_{\ell'}^{s_{\ell'}} = a$. The important observation is the following. Let $\pi$ be any random order that causes $\mathcal{E}_S$ and $\prod_{\ell' \in I'} X_{\ell'}^{s_{\ell'}} = a$. Then any random order $\pi'$ which differs only in positions $1, \ldots, \ell'_1$ also yields $\prod_{\ell' \in I'} X_{\ell'}^{s_{\ell'}} = a$ 
as the exact order of the requests $1, \ldots, \ell'_1$ is irrelevant for the LP solution computed in rounds $\ell' > \ell'_1$. 

In other words, to bound $\Ex{X_{\ell'_1} \growingmid \mathcal{E}_S \wedge \prod_{\ell' \in I'} X_{\ell'}^{s_{\ell'}} = a }$, we may still consider the order of the requests $1, \ldots, \ell'_1$ as unbiased. Formally, let $S' \subseteq [n]$, $\lvert S' \vert = \ell'_1$, be any set of $\ell'_1$ requests and let $\mathcal{E}_{S'}$ be the event that the requests in $S'$ come within the first $\ell'_1$ steps of the random input order. Then, conditioned on $\mathcal{E}_{S'} \wedge \mathcal{E}_S \wedge \prod_{\ell' \in I'} X_{\ell'}^{s_{\ell'}} = a$, each request from $S'$ comes at position with $\ell'_1$ with probability $\frac{1}{\ell'_1}$. In other words, it can be considered uniformly drawn from $S'$. Given a fixed event $\mathcal{E}_{S'}$, the LP solution $\tilde{x}^{(\ell'_1)}$ computed in round $\ell'_1$ is fixed. By definition $(A\tilde{x}^{(\ell'_1)})_i \leq \frac{\ell'_1}{n} b_i$. Therefore we know that, even in the conditional probability space of $\mathcal{E}_{S'}$, $\mathcal{E}_S$, and $\prod_{\ell' \in I'} X_{\ell'}^{s_{\ell'}} = a$, the expected contribution of $x^{(\ell'_1)}$ to constraint $i$ is bounded by $\frac{1}{\ell'_1} \frac{\ell'_1}{n} b_i$. So we get
\[
\Ex{X_{\ell'_1} \growingmid  \mathcal{E}_S \wedge \prod_{\ell' \in I'} X_{\ell'}^{s_{\ell'}} } 
= \frac{1}{\ell'_1} \Ex{(A\tilde{x}^{(\ell'_1)})_i \growingmid \mathcal{E}_S \wedge \prod_{\ell' \in I'} X_{\ell'}^{s_{\ell'}} } 
\leq \frac{1}{\ell'_1} \cdot \frac{\ell'_1}{n}b_i = \frac{b_i}{n} \enspace.
\]
As $X_{\ell'_1} \in [0, 1]$ and $s_{\ell'_1} \geq 1$, this bound also holds for $X_{\ell'_1}^{s_{\ell'_1}}$. Therefore, we get
\[
\Ex{\prod_{\ell' \in I} X_{\ell'}^{s_{\ell'}} \growingmid \mathcal{E}_S} = \Ex{X_{\ell'_1}^{s_{\ell'_1}} \prod_{\ell' \in I'} X_{\ell'}^{s_{\ell'}} \growingmid \mathcal{E}_S } \leq \frac{b_i}{n} \Ex{\prod_{\ell' \in I'} X_{\ell'}^{s_{\ell'}} \growingmid \mathcal{E}_S} \leq \frac{b_i}{n} \prod_{\ell' \in I'} \frac{b_i}{n} = \prod_{\ell' \in I} \frac{b_i}{n} \enspace,
\]
which shows that the random variables $X_{\ell'}$ are $1$-correlated.

Let us now apply a Chernoff bound on $\sum_{\ell' < \ell} X_{\ell'}$. First of all, we observe
\[
\Ex{X_{\ell'}} \leq \Ex{\hat{X}_{\ell'}} \quad \text{ and } \quad \Ex{\sum_{\ell' < \ell} \hat{X}_{\ell'} \growingmid \mathcal{E}_S } \leq \sum_{\ell' = 1}^{\ell - 1} \frac{b_i}{n} \leq \frac{\ell}{n} b_i \enspace.
\]
Now set $\delta = 1 - \frac{\ell}{n} \frac{b_i}{b_i - 1}$ so that $(1 - \delta) (b_i - 1) = \frac{\ell}{n} b_i \geq \Ex{\sum_{\ell' < \ell} \hat{X}_{\ell'} \growingmid \mathcal{E}_S }$. 
By the bound on $\ell$, we have $\frac{1}{b_i - 1} \leq \frac{1}{\sqrt{b_i}} \leq \frac{1}{9} \left( 1 - \frac{\ell}{n} \right)$, which implies 
\begin{equation}
\delta = \left(1 - \frac{\ell}{n}\right) - \frac{\ell}{n} \cdot \frac{1}{b_i - 1} \geq \frac{8}{9} \left( 1 - \frac{\ell}{n} \right) \enspace.
\label{eq:deltabound}
\end{equation}
On the other hand we have $1 - \delta \geq \frac{\ell}{n}$
and so we get 
$\min\left\{ \delta, 1 - \delta \right\} \geq \min\left\{ \frac{8}{9} ( 1 - \frac{\ell}{n} ), \frac{\ell}{n} \right\} \geq 8 \sqrt{\frac{1 + \ln d}{b_{\min}}}$.
Using furthermore the fact that $\delta (1 - \delta) \geq \frac{1}{2} \min\left\{ \delta, 1 - \delta \right\}$, we obtain
\[
\frac{b_i - 1}{3} \delta (1 - \delta) \geq \frac{b_i - 1}{3} \cdot 4 \sqrt{\frac{1 + \ln d}{b_{\min}}} 
\enspace.
\]
By multiplying this inequality with \eqref{eq:deltabound}, we get
\[
\frac{\delta^2}{3}  (1 - \delta) (b_i - 1) \geq  \frac{8}{9} \left( 1 - \frac{\ell}{n} \right) \frac{4(b_i - 1)}{3} \sqrt{\frac{ 1 + \ln d }{b_{\min}}} \geq \left( 1 + \frac{1}{9} \right) \left( 1 - \frac{\ell}{n} \right) \sqrt{b_i \left( 1 + \ln d \right)} \geq \ln d + \frac{n - \ell}{n} \sqrt{b_i}
\enspace,
\]
where we used $\left( 1 - \frac{\ell}{n} \right) \sqrt{b_i} \geq 9 \sqrt{\ln d}$ in the last step.

Putting the pieces together, we get that
\begin{align*}
& \Pr{\sum_{\ell' < \ell} X_{\ell'} > b_i - 1 \growingmid \mathcal{E}_S } \leq \Pr{\sum_{\ell' < \ell} X_{\ell'} > (1+\delta)(1 - \delta) (b_i - 1) \growingmid \mathcal{E}_S }\\
& \leq \exp\left( - \frac{\delta^2}{3} (1 - \delta) (b_i - 1) \right) 
\leq \exp\left( - \ln d - \frac{n - \ell}{n} \sqrt{b_i} \right) = \frac{1}{d} \exp\left( - \frac{n - \ell}{n} \sqrt{b_i} \right) \enspace.
\end{align*}
\end{proof}

Based on this lemma, we can now bound the expected value of the allocation $y$ computed by the algorithm. Again, it will be crucial to carefully keep track of dependencies.

\begin{theorem}\label{theorem:high-capacity}
When every column of the linear program has at most $d$ non-zero entries then Algorithm~\ref{alg:packing_algorithm} achieves a $1 - O\left( \sqrt{\frac{\log d}{b_{\min}}}\right)$-approximation even with respect to the optimal fractional solution.
\end{theorem}

\begin{proof}
We bound the expected value that the algorithm obtains by summing over all iterations. Let us first consider a fixed iteration $\ell$ with $pn \leq \ell \leq (1-p)n$ for $p = 9 \sqrt{\nicefrac{(1 + \ln d)}{b_{\min}}}$. Let $y^{(\ell)}$ be the change of the allocation in this round. To bound the expected improvement of the objective function $c^T y^{(\ell)}$, it will be helpful to think of the random order with respect to rounds $1, \ldots, \ell$ being determined in three steps:
\begin{enumerate}
\item[(i)] First, it is determined, which requests come within the first $\ell$ rounds (but not their order).
\item[(ii)] Second, one of these requests is selected to come in round $\ell$.
\item[(iii)] Finally, the order among the first $\ell - 1$ requests is determined.
\end{enumerate}
Observe that after (i) the LP solution $\tilde{x}^{(\ell)}$ is already determined. After (ii), $x^{(\ell)}$ is fixed as well. Step (iii) finally determines whether the allocation can actually be carried out, that is, if $y^{(\ell)} = x^{(\ell)}$.

By Lemma~\ref{lemma:local_solution_value}, we know that $\Ex{c^T \tilde{x}^{(\ell)}} \geq \frac{\ell}{n}  \left(1 - 9\sqrt{\frac{1 + \ln d}{\frac{\ell}{n} b_{\min}}} \right) \OPT$.

Step (ii) can actually be considered as selecting one of the first $\ell$ requests uniformly at random. Therefore we have $\Ex{c^{T}x^{(\ell)}} \geq \frac{1}{\ell}\cdot \Ex{c^{T}\tilde{x}^{(\ell)}}$.

In Lemma~\ref{lemma:allocation_probability}, we have shown that, independent of the outcomes of steps (i) and (ii), any constraint $i$ has remaining capacity less than one with probability at most $\frac{1}{d} \exp\left( - \frac{n - \ell}{n} \sqrt{b_{\min}} \right)$. Taking a union bound over all constraints having a non-zero entry, we observe that the probability that the allocation can be carried out is at least $1 - \exp\left( - \frac{n - \ell}{n} \sqrt{b_{\min}} \right)$. Formally, this means
\begin{align*}
\Ex{c^T y^{(\ell)}} &\geq \left( 1 - \exp\left( - \frac{n - \ell}{n} \sqrt{b_{\min}} \right) \right) \Ex{c^T x^{(\ell)}} \\
&\geq \left( 1 - \exp\left( - \frac{n - \ell}{n} \sqrt{b_{\min}} \right) \right) \frac{1}{n}  \left(1 - 9\sqrt{\frac{1 + \ln d}{\frac{\ell}{n} b_{\min}}} \right) \OPT  \enspace.
\end{align*}
Summing up all rounds $\ell$ and simplifying the expression, we get
\begin{align*}
\Ex{\ALG} &\geq \sum_{\ell = pn}^{(1-p)n} \frac{1}{n}  \left( 1 - \exp\left( - \frac{n - \ell}{n} \sqrt{b_{\min}}\right) \right) \left(1 - 9\sqrt{\frac{1 + \ln d}{\frac{\ell}{n} b_{\min}}} \right) \OPT \\
&\geq \left( 1 - 2p - \sum_{\ell = pn}^{(1-p)n} \frac{1}{n} \exp\left( - \frac{n - \ell}{n} \sqrt{b_{\min}}\right) - \sum_{\ell = pn}^{(1-p)n} \frac{1}{n} 9\sqrt{\frac{1+ \ln d}{\frac{\ell}{n} b_{\min}}} \right) \OPT \enspace.
\end{align*}
We will analyze both negative sums separately and bound them by multiples of $p$.

For the first sum we reverse the order of summation and extend the sum to a geometric series. Also using the inequality $1 - \exp(-x) \geq (1 - \frac{1}{\e}) x$ for all $0 \leq x \leq 1$ we obtain
\[
\frac{1}{n} \sum_{\ell = pn}^{(1-p)n} \exp \left( - \frac{n - \ell}{n}  \sqrt{b_{\min}} \right) 
\leq \frac{1}{n} \sum_{i = 0}^{\infty} \exp \left( - \frac{i}{n}  \sqrt{b_{\min}} \right)
\leq \frac{1}{n \left( 1 - \exp \left( - \frac{\sqrt{b_{\min}}}{n} \right) \right)}  \leq \frac{1}{\left( 1 - \frac{1}{\e} \right)  \sqrt{b_{\min}}} \leq p\enspace .
\]

For the second sum we have
\[
\sum_{\ell = pn}^{(1-p)n} \frac{1}{n} 9\sqrt{\frac{1 + \ln d}{\frac{\ell}{n} b_{\min}}} \leq 9\sqrt{\frac{1 + \ln d}{n b_{\min}}} \sum_{\ell = 1}^n \frac{1}{\sqrt{\ell}} \leq 18 \sqrt{\frac{1 + \ln d}{b_{\min}}} \leq 2p \enspace.
\]

As a result we get
\[
\Ex{\ALG} \geq (1 - 5p)\OPT = \left(1-O\left(\sqrt{\frac{1 + \log d}{b_{\min}}}\right)\right)\OPT \enspace.
\]
\end{proof}

%% file: lp_proof_small.tex

For small $b_{\min}$ there is no $\epsilon < 1$ so that the algorithm is $(1-\epsilon)$-competitive.
Fortunately we can still give a non-trivial bound on the performance of Algorithm~\ref{alg:packing_algorithm} for $b_{\min} \geq 2$ with a proof that is analogous to the previous section.
Again we start the analysis with a bound on the expected value of the scaled LP. 
In step $\ell$ a fraction of $\frac{\ell}{n}$ columns has already arrived and the capacities are scaled down by $\frac{\ell}{n}$ leading directly to the observation.
\begin{observation} \label{obs:simple_local_value}
Let $S \subseteq [n]$ be a random subset of requests with $\lvert S \rvert = \ell$. 
Then we have
\[
\Ex{\max_{x \in \mathcal{P}(\frac{\ell}{n}, S)} c^T x} 
\geq \left(\frac{\ell}{n}\right)^2 \cdot \max_{x \in \mathcal{P}(1, [n])} c^T x \enspace .
\]
\end{observation}
Next we bound the failure probability for an allocation and apply a union bound over all relevant online steps.

\begin{theorem}
\label{theorem:small-capacity}
When every column of the linear program has at most $d$ non-zero entries and we have $b_{\min} \geq 2$, then Algorithm~\ref{alg:packing_algorithm} is $\Omega\left(\frac{1}{d^{\nicefrac{2}{(b_{\min} - 1)}}}\right)$-competitive.
\end{theorem}

\begin{proof}
Let $\OPT$ be the value of the global optimum and set $\psi = d^{\frac{1}{b_{\min} - 1}}$.
By Observation~\ref{obs:simple_local_value} we have for all rounds $\ell \geq \frac{n}{8\e\psi}$ that
\[
\Ex{c^T x^{(\ell)}} \geq \frac{1}{\ell} \left(\frac{\ell}{n}\right)^2 \OPT \geq \frac{1}{8\e\psi} \frac{\OPT}{n} \enspace.
\]
Let us consider the tentative allocation $x^{(\ell)}$ in any fixed round $\ell \leq \frac{n}{4\e\psi}$, where request $j$ arrives.
This assignment can be carried out if $(\sum_{\ell' < \ell} A x^{(\ell')})_i \leq b_i - 1$ for all $i \in [m]$ with $a_{i,j} > 0$.
We have for all $i \in [m]$
\[
\Ex{\left(\sum_{\ell' < \ell} A x^{(\ell')}\right)_i} \leq \sum_{\ell' = 1}^{\ell - 1} \frac{1}{\ell'} \frac{\ell'}{n} b_i = \frac{\ell - 1}{n} b_i \leq \frac{b_i}{4\e\psi} \enspace.
\]
Define $\delta = 4\e\psi \left(1 -  \frac{1}{b_i} \right)- 1$. We have $(1 + \delta) \frac{1}{4\e\psi} b_i = b_i - 1$. By the same reasons as discussed in the proof of Lemma~\ref{lemma:allocation_probability}, we can apply a Chernoff bound and get
\begin{align*}
\Pr{\left(\sum_{\ell' < \ell} A x^{(\ell')}\right)_i \geq b_i - 1} & = \Pr{\left(\sum_{\ell' < \ell} A x^{(\ell')}\right)_i \geq (1 + \delta) \frac{b_i}{4\e\psi}} \\
& \leq \left( \frac{\e^\delta}{(1 + \delta)^{1 + \delta}} \right)^{\frac{b_i}{4\e\psi}} 
\leq \left( \frac{\e}{1 + \delta} \right)^{(1 + \delta) \frac{b_i}{4\e\psi}} \enspace .
\end{align*}
Furthermore, since $b_i \geq 2$, we have $1 + \delta \geq \frac{1}{2} 4\e\psi = 2 \e \psi$ and therefore 
\[
\left( \frac{\e}{1 + \delta} \right)^{(1 + \delta) \frac{b_i}{4\e\psi}} 
\leq \left( \frac{1}{2\psi}  \right)^{b_i - 1} 
\leq \frac{1}{2d} \enspace .
\]
Using a union bound, we observe that this tentative assignment $x^{(\ell)}$ is carried out with probability at least $\frac{1}{2}$. That is, we get
\[
\Ex{\ALG} \geq \sum_{\ell = \frac{n}{8\e\psi}}^{\frac{n}{4\e\psi}} \frac{1}{2} \Ex{c^T x^{(\ell)}} 
\geq \left( \frac{n}{4\e\psi} - \frac{n}{8\e\psi} \right) \frac{1}{2} \frac{1}{8\e\psi} \frac{\OPT}{n} 
= \Omega\left(\frac{1}{d^{\nicefrac{2}{(b_{\min} - 1)}}}\right) \OPT \enspace.
\]
\end{proof}

%% file: truthful.tex

A common scenario for online packing LPs are auctions, in which bidders arrive one by one and need to be served immediately. This means, upon arrival each bidder reports her options and respective valuations.
Based on this and previously gathered information, the allocation for this specific bidder and a payment are determined.
While the previously presented algorithm already solves the underlying optimization problem, in an auction setting bidders might strategically misreport.
This strategic behavior can undermine the performance guarantee.
An algorithm is robust against this type of manipulation if it is in each bidder's best interest to report the truth, the algorithm is said to be \emph{truthful}.

In Algorithm~\ref{alg:packing_algorithm}, if a resource has not enough capacity left to fully allocate an option of the current bidder, she is potentially better off by not reporting any valuation for this option. Algorithm~\ref{alg:truthful} deals with this issue by discarding all options that are not feasibly satisfiable. Afterwards it continues in the same way as Algorithm~\ref{alg:packing_algorithm} and applies VCG payments. We observe that this mechanism is truthful in expectation, even when bidders know the random order in advance. This is due to the fact that the \emph{supporting mechanism} is truthful because it always computes the social-welfare optimum and applies VCG; for details see~\cite{DBLP:journals/jacm/LaviS11}.

\begin{algorithm}
\caption{Truthful online packing LP\label{alg:truthful}}
\SetKwInOut{Input}{Input} \SetKwInOut{Output}{Output}
Let $S$ be the index set of known requests, initially $S := \emptyset$\;
Set $y := 0$\;
\For(\tcp*[f]{steps $\ell = 1$ to $n$}){\emph{each arriving request} $j$}{
  Set $S := S \cup \{j\}$ and $\ell := |S|$\;  
  \For{each option $k$ such that $\exists i \in [m]$ with $(Ay)_i + a_{i,j,k} > b_i$}{
	Set $c_{j,k} = 0$; \tcp*[f]{remove infeasible options}\\
  }
  Let $\tilde{x}^{(\ell)}$ be an optimal solution of the scaled LP $\max_{x \in \mathcal{P}(\frac{\ell}{n}, S)} c^T x$\;
  Charge payment $p_j = \max_{x \in \mathcal{P}(\frac{\ell}{n}, S \setminus\{j\})} c^T x - \sum\limits_{j' \in S\setminus \{ j \}, k\in[K]} c_{j', k} \cdot \tilde{x}^{(\ell)}_{j',k}$ \tcp*[j]{VCG payment}
  Choose an option $k^{(\ell)}$ (possibly none) where option $k$ has probability $\tilde{x}^{(\ell)}_{j, k}$; \tcp*[f]{rand.\ rounding}\\
  Define $x^{(\ell)}$ with $x^{(\ell)}_{j', k} = \begin{cases} 1 , & \text{if } j' = j \text{ and } k = k^{(\ell)};  \\
                                                               0 , & \text{otherwise} ;\end{cases}$ \\
    Set $y := y+x^{(\ell)}$; \tcp*[f]{permanent online allocation}\\
}
\end{algorithm}

\begin{theorem}
Given truthful reports, the mechanism approximates optimal social welfare within a factor of $1 - O\left( \sqrt{\frac{\log m}{b_{\min}}}\right)$.
\end{theorem}

\begin{proof}
Like in the proof of Theorem~\ref{theorem:high-capacity}, we consider the contribution to the objective function of a fixed iteration $\ell$ with $pn \leq \ell \leq (1-p)n$ for $p = 9 \sqrt{\nicefrac{(1 + \ln m)}{b_{\min}}}$. Let $\mathcal{F}$ be the event that $(\sum_{\ell' < \ell} A x^{(\ell')})_i \leq b_i - 1$ for all $i \in [m]$. Observe that as long as $\mathcal{F}$ holds the algorithm's behavior in round $\ell$ does not differ from the one of Algorithm~\ref{alg:packing_algorithm}. Let $z^{(\ell)}$ be the random variable indicating the tentative allocation in a run of Algorithm~\ref{alg:packing_algorithm}.

For all $a > 0$ we now have $\Pr{c^T x^{(\ell)} \geq a} \geq \Pr{\mathcal{F}} \Pr{c^T x^{(\ell)} \geq a \growingmid \mathcal{F}} = \Pr{\mathcal{F}} \Pr{c^T z^{(\ell)} \geq a \growingmid \mathcal{F}} = \Pr{c^T z^{(\ell)} \geq a} \Pr{\mathcal{F} \growingmid c^T z^{(\ell)} \geq a}$.

In the proof of Theorem~\ref{theorem:high-capacity}, we have shown that $\Ex{c^T z^{(\ell)}} \geq \frac{1}{n} \left(1 - 9\sqrt{\frac{1 + \ln m}{\frac{\ell}{n} b_{\min}}} \right) \OPT$. Furthermore, Lemma~\ref{lemma:allocation_probability} shows that $\Pr{\mathcal{F} \growingmid c^T z^{(\ell)} \geq a} \geq 1 - \exp\left( - \frac{n - \ell}{n} \sqrt{b_{\min}} \right)$.
In combination, we get $\Ex{c^T x^{(\ell)}} \geq \frac{\OPT}{n}  \left(1 - 9\sqrt{\frac{1 + \ln m}{\frac{\ell}{n} b_{\min}}} \right) \left(1 - \exp\left( - \frac{n - \ell}{n} \sqrt{b_{\min}} \right)\right)$.

The remaining calculations can be carried out exactly as in the proof of Theorem~\ref{theorem:high-capacity}.
\end{proof}
By an analogous adaptation of Theorem~\ref{theorem:small-capacity}, we can also show a performance bound of $\Omega\left(\frac{1}{m^{\nicefrac{2}{(b_{\min} - 1)}}}\right)$.

%% file: small_tailored.tex

We have shown performance guarantees for Algorithm~\ref{alg:packing_algorithm} for any $B \geq 2$. The algorithm has the advantage that it does not need to know $B$ or $d$ in advance. However, in case of small $B$, the allocation might be somewhat too optimistic and resources can be exhausted early. Therefore, if we know $B$ and $d$, we can achieve a better competitive ratio by adding a sampling phase at the start. In more detail, we run the same algorithm but do not make any allocation in the first $pn$ rounds, where $p = 1 - \frac{1}{2 \e} \left( \frac{1}{2d} \right)^{\frac{1}{B - 1}}$. 

\begin{lemma}
The modified algorithm is $\Omega\left(\frac{1}{d^{\nicefrac{1}{(B - 1)}}}\right)$-competitive.
\end{lemma}

\begin{proof}
We can bound $\Ex{c^T x^{(\ell)}}$ again by using Observation~\ref{obs:simple_local_value}. Additionally, we can use $\ell \geq p n$ and $p \geq \frac{1}{2}$ to get
\[
\Ex{c^T x^{(\ell)}} \geq \frac{1}{\ell} \left( \frac{\ell}{n} \right)^2 \OPT \geq \frac{1}{2} \frac{\OPT}{n} \enspace.
\]

We have for all $i \in [m]$
\[
\Ex{\left(\sum_{\ell' < \ell} A x^{(\ell')}\right)_i} \leq \sum_{\ell' = pn + 1}^{\ell - 1} \frac{1}{\ell'} \frac{\ell'}{n} b_i \leq \frac{\ell - pn}{n} b_i \leq (1 - p) b_i \enspace.
\]
Defining $\delta = \frac{1 - \nicefrac{1}{b_i}}{1 - p} - 1$ yields $(1 + \delta) (1 - p) b_i = b_i - 1$. Like in the proofs of Lemma~\ref{lemma:allocation_probability} and Theorem~\ref{theorem:small-capacity}, we can apply a Chernoff bound to get
\begin{align*}
\Pr{\left(\sum_{\ell' < \ell} A x^{(\ell')}\right)_i \geq b_i - 1} & = \Pr{\left(\sum_{\ell' < \ell} A x^{(\ell')}\right)_i \geq (1 + \delta) (1 - p) b_i}  \leq \left( \frac{\e^\delta}{(1 + \delta)^{1 + \delta}} \right)^{(1-p) b_i} \\
& \leq \left( \frac{\e}{1 + \delta} \right)^{b_i - 1} = \left( \frac{\e (1-p)}{1 - \frac{1}{b_i}} \right)^{b_i - 1} \leq \left( 2 \e (1-p) \right)^{b_i - 1} = \frac{1}{2d}
\end{align*}
A union bound shows that the tentative assignment can be carried out with probability at least $\frac{1}{2}$.
That is, we get
\[
\Ex{\ALG} \geq \sum_{\ell = pn + 1}^n \frac{1}{2} \Ex{c^T x^{(\ell)}} 
\geq (1-p)n \frac{1}{4} \frac{\OPT}{n} 
= \Omega\left(\frac{1}{d^{\nicefrac{1}{(b_{\min} - 1)}}}\right) \OPT \enspace.
\]
\end{proof}

The analogous modification can also be made to Algorithm~\ref{alg:truthful} with $p = 1 - \frac{1}{2 \e} \left( \frac{1}{2m} \right)^{\frac{1}{B - 1}}$. Incentive compatibility is preserved and we achieve a competitive ratio of $\Omega\left(\frac{1}{m^{\nicefrac{1}{(B - 1)}}}\right)$.

%% file: gap_proof.tex

The special case of linear packing programs with $d=1$ is identical to the generalized assignment problem (GAP).
The classical definition assumes $m$ bins or resources, where bin $i$ has capacity $b_i$.
Additionally, there are $n$ items which may be placed into the bins.
Depending on the bin an item is assigned to, it has a specific size and raises a specific profit.
In particular, if item $j$ is assigned to bin $i$ it consumes $w_{i, j} \geq 0$ units of the bins capacity and raises a profit of $p_{i, j} \geq 0$.
The objective is to maximize the total assigned profit while not exceeding the capacities of the bins.
\begin{align} \label{linearprogram:gap}
 \max          \quad \sum_{i \in [m],\ j \in [n]} p_{i, j} \cdot &x_{i, j} \\
 \mbox{s.\,t.} \quad \sum_{j \in [n]}             w_{i, j} \cdot x_{i, j} &\leq b_i     &i\in [m] \nonumber \\
                     \sum_{i \in [m]}                            x_{i, j} &\leq 1       &j\in [n] \nonumber \\
                                                           x_{i, j} \in \{&0, 1\}       &i\in [m],\ j\in [n] \nonumber                                   
\end{align}

Well-known special cases of GAP are edge-weighted matching ($\forall i, \forall j: w_{i, j} = 1, b_i = 1$), the knapsack problem ($m=1$), unweighted bipartite matching ($\forall i, \forall j: p_{i, j} = w_{i, j} = 1, b_i = 1$) and the AdWords problem ($\forall i, \forall j: p_{i, j} = w_{i, j}$).

In this section we will use the above LP-formulation which is identical to the notation used in the literature. The correspondence of this linear program to the model defined in Section~\ref{sec:definition} can easily be seen by assuming that every online item $j$ comes with multiple columns, one for every bin $i$ that it can be assigned to. Here, column $i$ of request $j$ has exactly one non-zero entry $a_{i, j, i} = w_{i, j}$ and the objective function value $c_{j, i}$ is $p_{i, j}$.

Our algorithm is based on the following simple observation. If all items consume more than half of a bins capacity then we can assign at most one item per bin. Hence, such an instance is identical to edge-weighted matching. Given a general GAP-instance we define the restricted instance $\mathcal{I}^{\text{heavy}}$ where we only allow those options with $w_{i, j} > \frac{1}{2} b_i$. The complementary restricted instance with the options $w_{i, j} \leq \frac{1}{2} b_i$ will be denoted by $\mathcal{I}^{\text{light}}$. Our algorithm will make a random choice whether to exclusively consider $\mathcal{I}^{\text{heavy}}$ or $\mathcal{I}^{\text{light}}$.

\begin{algorithm}
\caption{Online generalized assignment problem\label{alg:GAP_algorithm}}
\SetKwInOut{Input}{Input} \SetKwInOut{Output}{Output}
Flip a biased coin with $\Pr{\text{heads}} = \lambda$ and $\Pr{\text{tails}} = 1-\lambda$\;
\eIf{\text{heads}}{
Only consider options with $w_{i, j} > \frac{1}{2} b_i$\;
Use the online algorithm by Kesselheim et al.~\cite{DBLP:conf/esa/KesselheimRTV13} for edge-weighted matching\;
} {
Only consider options with $w_{i, j} \leq \frac{1}{2} b_i$\;
Let $S$ be the index set of the first $pn$ incoming items\;
  \For(\tcp*[f]{steps $\ell =  pn + 1$ to $n$}){\text{each subsequently arriving item} $j$}{
    Set $S := S \cup \{j\}$ and $\ell := |S|$\;
    Let $\tilde{x}^{(\ell)}$ be an optimal fractional solution of the current LP-relaxation (i.\,e.\ restricted to visible items $S$ and \emph{light} options)\;
    Choose a bin $i^{(\ell)}$ (possibly none), where bin $i$ has probability $\tilde{x}^{(\ell)}_{i, j}$\;
    \If{\text{item} $j$ (\text{now with} $w_{i^{(\ell)}, j}$) \text{still fits into bin} $i^{(\ell)}$}{
      Assign item $j$ to bin $i^{(\ell)}$\;
    }
  }
}
\end{algorithm}

\begin{theorem}
Choosing the parameter $\lambda = \frac{1}{1 + \nicefrac{16}{3 \e}}$ and $p = \frac{2}{3}$, Algorithm~\ref{alg:GAP_algorithm} is $\frac{1}{8.1}$-competitive. 
\end{theorem}

\begin{proof}
First we analyze the case when the coin shows heads.
In this situation the instance is restricted to $\mathcal{I}^{\text{heavy}}$ which has an optimal value of $\OPT^{\text{heavy}}$.
As noted above, this is an instance of edge-weighted matching.
Using the online algorithm by Kesselheim et al.~\cite{DBLP:conf/esa/KesselheimRTV13} we obtain an assignment that is $\frac{1}{\e}$-competitive in expectation with respect to $\OPT^{\text{heavy}}$, hence
\begin{equation} \label{GAP_inequality_heads}
 \Ex{ALG \mid \text{heads}} \geq \frac{1}{\e} \cdot \OPT^{\text{heavy}} \enspace .
\end{equation}

If the coin shows tails the algorithm considers $\mathcal{I}^{\text{light}}$ and works similarly to Algorithm~\ref{alg:packing_algorithm}, except for an additional sampling phase and unscaled capacities.
Fix a step $\ell \geq pn +1$ and let $x^*$ be on optimal fractional solution of $\mathcal{I}^{\text{light}}$ with objective value $\OPT^{\text{light}}$.
The set of visible items is a random subset $S \subseteq \mathcal{I}^{\text{light}}$ with cardinality $\ell$.
Consider the projected vector $x^\prime$ with $x^\prime_{i, \tilde{j}} = x^*_{i, \tilde{j}} $ if $\tilde{j} \in S$, and $x^\prime_{i, \tilde{j}} = 0$ otherwise.
Obviously, $\Ex{p^T \tilde{x}^{(\ell)}} \geq \Ex{p^T x^\prime} = \frac{\ell}{n} \cdot \OPT^{\text{light}}$.
Note that the online item $j$ can be seen as being uniformly chosen from the $\ell$ items in $S$.
The bin $i^{(\ell)}$ is determined by interpreting the fractional allocation of $j$ in $\tilde{x}^{(\ell)}$ as a probability distribution.
Hence, the expected profit of a tentative allocation in round $\ell$ is $\Ex{p_{i^{(\ell)}, j}} = \frac{1}{\ell} \cdot \Ex{p^T \tilde{x}^{(\ell)}} \geq \frac{1}{n} \cdot \OPT^{\text{light}}$.

Since $w_{i^{(\ell)}, j} \leq \frac{1}{2} b_{i^{(\ell)}}$, the allocation of item $j$ to bin $i^{(\ell)}$ will be successful if the previous resource consumption of bin $i^{(\ell)}$ is at most $\frac{1}{2}b_{i^{(\ell)}}$. In every previous round we solved the LP-relaxation of $\mathcal{I}^{\text{light}}$ restricted to the then visible items in $S$ and randomly rounded the fractional allocation. Again, since the then online item can be seen as being uniformly chosen from the visible items, the expected resource consumption of the tentative allocation in round $\ell' < \ell$ is at most $\frac{b_i}{\ell'} \ (\forall i)$. Hence, the expected consumption of any bin $i$ before round $\ell$ is at most $\sum_{\ell'=pn+1}^{\ell-1} \frac{b_i}{\ell'}$. Using Markov's inequality we get
$$\Pr{\text{allocation successful}} \geq 1 - \frac{\sum_{\ell'=pn+1}^{\ell-1} \nicefrac{b_{i^{(\ell)}}}{\ell'}}{\nicefrac{b_{i^{(\ell)}}}{2}} = 1 - 2 \cdot \sum_{\ell'=pn+1}^{\ell-1} \frac{1}{\ell'}\enspace.$$

Combining the expected profit of each tentative allocation with its success probability and summing over all rounds we can bound the expected profit of the algorithm:
\begin{align*}
 \Ex{ALG \mid \text{tails}} &\geq \sum_{\ell = pn +1}^{n} \frac{\OPT^{\text{light}}}{n} \cdot \left(1 - 2 \cdot \sum_{\ell'=pn+1}^{\ell-1} \frac{1}{\ell'} \right) \\
 &= \frac{\OPT^{\text{light}}}{n} \cdot \left((1-p)n - 2 \cdot \sum_{\ell' = pn +1}^{n}   \left( \frac{n}{\ell'} - 1\right) \right) \\
 &\geq \OPT^{\text{light}} \cdot \left(3(1-p) - 2 \ln{\left(\frac{1}{p}\right)} \right) \enspace .
\end{align*}
In the last inequality we used the fact $\sum_{\ell'=pn+1}^{n}\frac{1}{\ell'} < \int_{pn}^{n}\frac{1}{t}\text{d}t = \ln{(\nicefrac{1}{p})}$.
By the choice of the parameter $p=\nicefrac{2}{3}$ we have
\begin{equation} \label{GAP_inequality_tails}
 \Ex{ALG \mid \text{tails}} \geq \left(1 - \ln{\left(\frac{9}{4}\right)}\right) \cdot \OPT^{\text{light}} \geq \frac{3}{16} \cdot \OPT^{\text{light}}\enspace .
\end{equation}

Finally we can combine the two inequalities~(\ref{GAP_inequality_heads}) and~(\ref{GAP_inequality_tails}).
Together with $\lambda = \frac{1}{1 + \nicefrac{16}{3\e}}$ and since $\OPT^{\text{heavy}} + \OPT^{\text{light}} \geq \OPT$ we get 
\begin{align*}
 \Ex{ALG} &\geq \lambda \cdot \Ex{ALG \mid \text{heads}} + (1-\lambda) \cdot \Ex{ALG \mid \text{tails}} \\
 &\geq \frac{\lambda}{\e} \cdot \OPT^{\text{heavy}}  + \frac{(1-\lambda) 3}{16} \cdot \OPT^{\text{light}}
 = \frac{1}{\e + \nicefrac{16}{3}} (\OPT^{\text{heavy}} + \OPT^{\text{light}})
 \geq \frac{1}{8.1} \cdot \OPT \enspace .
\end{align*}
\end{proof}